%% file: draft.tex
\documentclass[reqno,a4paper]{amsart}
\input{definitions.tex}
\excludecomment{maximacode}
\excludecomment{arxivabstract}

\begin{document}

\title[Positive-Entropy Systems]{Positive-entropy Hamiltonian systems on Nilmanifolds via
  Scattering}
\author{Leo T. Butler}
\date{\today}
\thanks{The author thanks Adri Olde Daalhuis for his helpful comments
  on an early draft of this paper.}
\maketitle

\begin{abstract}
  Let $\Sigma$ be a compact quotient of $T_4$, the Lie group of $4
  \times 4$ upper triangular matrices with unity along the
  diagonal. The Lie algebra $\t_4$ of $T_4$ has the standard basis
  $\set{X_{ij}}$ of matrices with $0$ everywhere but in the $(i,j)$
  entry, which is unity. Let $g$ be the Carnot metric, a
  sub-riemannian metric, on $T_4$ for which $X_{i,i+1}$, $(i=1,2,3)$,
  is an orthonormal basis. Montgomery, Shapiro and Stolin showed that
  the geodesic flow of $g$ is algebraically non-integrable.

  This note proves that the geodesic flow of that Carnot metric on $T
  \Sigma$ has positive topological entropy and is real-analytically
  non-integrable. It extends earlier work by Butler and Gelfreich.
\end{abstract}

\section{Introduction} \label{sec:intro}

Let $G$ be a connected nilpotent Lie group with discrete subgroup $D$
and let $\Sigma=G/D$ be the corresponding homogeneous space. Each
homogeneous (sub-)riemannian metric on $G$ induces a
locally-homogeneous metric on $\Sigma$. These \textem{left-invariant}
geometries are interesting both geometrically and dynamically. A basic
question is

\begin{question}
  Which left-invariant geodesic flows on a compact nilmanifold have
  zero topological entropy?
\end{question}

Let $T_n$ be the nilpotent group of upper triangular $n \times n$ real
matrices with unity on the diagonal. Montgomery, Shapiro and
Stolin~\cite{MR1481625} investigate the geodesic flow of a Carnot
metric on $T_4$; they show that it reduces to the Yang-Mills
hamiltonian flow which is known to be algebraically
non-integrable~\cite{MR674006,MR695092}. In~\cite{MR2011121}, metrics
on compact quotients of the $3$-step nilpotent Lie group $T_4 \oplus
T_3$ are constructed whose geodesic flows have positive topological
entropy. In~\cite{MR2425326}, Butler \& Gelfreich showed that there
are riemannian and sub-riemannian metrics on $T_4$ which have positive
topological entropy. Numerical analysis in that paper suggested that
the Carnot metric of Montgomery, Shapiro and Stolin has a horseshoe,
hence positive topological entropy, and is analytically
non-integrable. This note proves those numerical results are, in fact,
correct. In that paper, a Melnikov integral is expressed in terms of
scattering data for a second-order scalar differential equation; in
the present paper, this scattering data is explicitly computed in
terms of $\Gamma$-functions.

The Lie algebra of $T_4$, $\t_4$, has the standard basis consisting of
those $4\times 4$ matrices $X_{ij}$ with a unit in the $i$-th row and
$j$-th column, $i < j$, and zeros everywhere else. We will restrict
attention to those structures $\ip{}{}$ where $\ip{X_{ij}}{X_{kl}} =
b_{ij}$ when $i=k, j=l$ and zero otherwise. The standard riemannian
metric has $b_{ij}=1$ for all $i,j$; the standard Carnot
sub-riemannian metric studied in \cite{MR1481625} has
$b_{12}=b_{23}=b_{34}=1$ and all other coefficients zero.

\begin{theorem}
  \label{thm:htop-pos}
  Let $\Sigma$ be a homogeneous space of $T_4$. The topological
  entropy of the geodesic flows of the standard riemannian and Carnot
  metrics is positive.
\end{theorem}

This theorem is proven by reducing the flows to hamiltonian flows on
$\t_4^*$, the dual of the Lie algebra $\t_4$ of $T_4$. The Poisson
sub-algebra of left-invariant hamiltonians on $T^* T_4$ is naturally
identified with the hamiltonians on $\t_4^*$ with the natural Poisson
structure. The Lie group's co-adjoint action is by Poisson
automorphisms and a co-adjoint orbit $\coorbit \subset \t_4^*$ is a
symplectic submanifold to which all such hamiltonian vector fields are
tangent. A quadratic hamiltonian $h : \t_4^* \to \R$ is \textem{diagonal}
if it is expressed as $h = \sum_{i<j} a_{ij} X_{ij}^2$ for some
constants $a_{ij}$.

\begin{theorem}
  \label{thm:main}
  Let $h : \t_4^* \to \R$ be a diagonal hamiltonian with
  $a_{13}a_{34}=a_{12}a_{24}$ and let $E_h$ be the hamiltonian, or
  Euler, vector field of $h$ on $\t_4^*$. There is an open set of
  regular co-adjoint orbits $\coorbit \subset \t_4^*$, such that $E_h
  | \coorbit$ has a horseshoe. In particular, the Euler vector field
  of the standard Carnot metric has positive topological entropy and
  is real analytically non-integrable.
\end{theorem}

Theorem \ref{thm:main} is proven by expressing a Melnikov integral as
a quadratic form in $2$-variables with coefficients that are obtained
by solving a scattering problem; these coefficients are naturally
expressed in terms of $\Gamma$-functions involving a parameter, called
$\alpha$ below, that depends on the metric coefficients $a_{ij}$ and
the co-adjoint orbit. Note that \cite{MR2425326} asserts that the
horseshoe exists on all but countably many regular co-adjoint orbits;
this is inaccurate. That paper shows the horseshoe exists for all but
countably many \textem{real} values of the invariant $\alpha$;
however, $\alpha$ may be \textem{imaginary} on an open set. This is
explained in figure \ref{fig:saddle-center} below. As noted in
\cite{MR2425326}, when $a_{13} = 0$, the invariant $\alpha$ is
constant and one cannot therefore conclude that there is a horseshoe
on any of the co-adjoint orbits. The standard Carnot metric of
\cite{MR1481625} falls into this case ($\alpha = 1$). The present
paper uses an alternative approach that shows the existence of a
horseshoe for all non-zero real values of $\alpha$. This is strong
enough to prove the existence of a horsehoe on an open set of
co-adjoint orbits, even when $a_{13}$ vanishes. It remains an open
question if the Euler vector field has a horseshoe on a co-adjoint
orbit where $\alpha$ is imaginary.

Theorem \ref{thm:main} implies, from the structural stability of the
horseshoe, that there is an open set $W$ of quadratic hamiltonians on
$\t_4^*$ each of which has a horseshoe; further, the $\Aut{T_4}$ orbit
of $W$ has this property, too.  This motivates the following:

\begin{question}
  Does there exist a quadratic hamiltonian on $\t_4^*$ which induces a
  non-degenerate (sub-)riemannian structure on $T_4$ and which is
  completely integrable or has zero topological entropy?
\end{question}

If one drops the non-degeneracy condition, then the answer is
trivially yes to both questions, as witnessed by $h=X_{14}^2$, which
is a Casimir.

\subsection{Outline}
\label{sec:outline}

This note is organized as follows: section \ref{sec:bg} reviews the
derivation of the Melnikov form from \cite{MR2425326}; section
\ref{sec:improper-integral} computes the integrals that arise in the Melnikov
form in terms of the scattering matrices at $\pm\infty$ in a general
scattering problem; section \ref{sec:scatter+split} demonstrates the
non-degeneracy of the Melnikov form for the particular form arising
from section \ref{sec:bg} and completes the proof of theorems
\ref{thm:htop-pos} and \ref{thm:main}.

\section{Background} \label{sec:bg}

This section recalls a number of facts about left-invariant
hamiltonian systems on the cotangent bundle of a Lie group;
see also~\cite{MR1723696,MR2425326}.

\subsection{Poisson geometry of left-invariant hamiltonians}

A \textem{Poisson manifold} is a smooth manifold $M$ such
that $C^{\infty}(M)$ is equipped with a skew-symmetric bracket $\{,\}$
that makes $( C^{\infty}(M), \{,\} )$ into a Lie algebra of
derivations of $C^{\infty}(M)$. The centre of $( C^{\infty}(M), \{,\}
)$ is the set of \textem{Casimirs}. A Casimir is a first integral of all
hamiltonian vector fields.

The dual space of a Lie algebra gives an example of a Poisson manifold
that is not (in general) a symplectic manifold. Let $\g$ be a
finite-dimensional real Lie algebra and let $\g^*$ be the dual vector
space of $\g$. $T^*_p \g^*$ is identified with $\g$ for all $p \in
\g^*$. The Poisson bracket on $\g^*$ is defined for all $f,h \in
C^{\infty}(\g^*)$ and $p \in \g^*$ by
\begin{equation} \label{eq:pb}
\pb{f}{h}(p) := - \ip{p}{ [ df_p, dh_p ] },
\end{equation}
where $\ip{}{} : \g^* \times \g \to \R$ is the
natural pairing. Recall that for $\xi \in \g$, $\ad{\xi}^* : \g^* \to
\g^*$ is the linear map defined by $\ip{\ad{\xi}^*p}{ \eta }
= -\ip{p}{ [\xi,\eta] }$. $\ad{}^* : \g \to gl(\g^*)$ is the
co-adjoint representation. For any $h \in C^{\infty}(\g^*)$, the
hamiltonian vector field $E_h = \pb{}{h}$ equals $-\ad{dh(p)}^*
p$.

Let $G$ be a connected Lie group whose Lie algebra is $\g$. The
adjoint representation of $G$ on $\g$, $\Ad{g}\xi =
\frac{d}{dt}|_{t=0}\, g \exp(t\xi) g^{-1}$, induces the co-adjoint
representation $\ip{\Ad{g}^* p}{ \xi } = \ip{p}{\Ad{g^{-1}} \xi }$ for
all $p \in \g^*$, $g \in G$ and $\xi \in \g$. As each vector field $p
\to \ad{\xi}^* p$ is hamiltonian on $\g^*$, with linear hamiltonian
$h_{\xi}(p) = - \ip{p}{ \xi }$, the co-adjoint action of $G$
on $\g^*$ preserves the Poisson bracket. The orbits of the co-adjoint
action are called the co-adjoint orbits. Each co-adjoint orbit is a
homogeneous $G$-space, and \textem{every} hamiltonian vector field on
$\g^*$ is tangent to each co-adjoint orbit. For this reason, the
Poisson bracket $\pb{}{}_{\g^*}$ restricts to each co-adjoint orbit,
and is non-degenerate on each co-adjoint orbit. Thus, the co-adjoint
orbits are naturally symplectic manifolds. A Casimir is necessarily
constant on each co-adjoint orbit, and in many cases (as in this paper)
each co-adjoint orbit is the common level set of all Casimirs.

The hamiltonian flow of a left-invariant hamiltonian $H$ on $T^* G$
has the equations of motion:
\begin{equation}
X_H(g,p) = \left\{ 
\begin{array}{ccc} 
\dot{g} &=& T_e L_g dh(p),\\
\dot{p} &=& -\ad{dh(p)}^*\ p,
\end{array}
\right.
\end{equation}
The reduction of the vector field $X_H$ to $\g^*$ is the Euler vector
field $E_h$.

\subsection{Poisson geometry of $T^*T_4$}

The Lie algebra of $T_4$ is
$$
\t_4 = \left\{ {
\left[ \begin{array}{cccc}
0 & x & z & w\\
0 & 0 & y & u\\
0 & 0 & 0 & v\\
0 & 0 & 0 & 0
\end{array} \right] } 
\ :\ u,v,w,x,y,z \in \R \right\}.
$$
Let $p_{\bullet}$ be the coordinate functions on $\t^*_4$ dual to the
above coordinates on $\t_4$. The Poisson bracket on $\t^*_4$ is:
\begin{align}
  \label{eq:pb-relns}
  \pb{p_x}{p_y}&=-p_z, &&& \pb{p_x}{p_u}=-p_w, \notag\\
  \pb{p_y}{p_v}&=-p_u, &&& \pb{p_z}{p_v}=-p_w.
\end{align}
There are two independent Casimirs of $\t_4^*$ are $K_1(p) = p_w$,
$K_2(p) = p_w p_y - p_z p_u$. Let $K : \t_4^* \to \R^2$ be defined by
$K = (K_1,K_2)$. The level sets of $K$ are the co-adjoint orbits of
$T_4$'s action on $\t_4^*$ and will be denoted by $\coorbit_{k}$,
where $k=(k_1,k_2)$. We will say that $\coorbit_{k}$ is a
\textem{regular co-adjoint orbit} if $k_1 \neq 0$.

\begin{lemma}
  \label{lem:ok}
  Each regular co-adjoint orbit $\coorbit_{k}$ is symplectomorphic
  to $T^* \R^{2}$ equipped with its canonical symplectic structure.
\end{lemma}

\begin{proof}
  Indeed, the right-hand column of the commutation relations
  \eqref{eq:pb-relns} shows that when $k_1=p_w \neq 0$, the
  coordinates $(p_x,p_u,p_z,p_v)$ are conformally symplectic and the
  first column is a consequence of $K_2=k_2$ constant on $\coorbit_k$. See \cite{MR2425326}.
\end{proof}

\subsection{The hamiltonians}

Let $a_{ij} \geq 0$ be constants such that $a_{12}a_{23}a_{34} \neq 0$
and $a_{13}a_{34}=a_{12}a_{24}$. Define
\begin{equation} \label{eq:h}
4h(p) = a_{12} p_x^2 + a_{23} p_y^2 + a_{13} p_z^2 + a_{24} p_u^2 +
a_{34} p_v^2 + a_{14} p_w^2
\end{equation}
As shown in \cite{MR2425326}, there is a change of coordinates that
transforms $h|\coorbit_k$ to
\begin{equation} \label{eq:h1}
2{\bf h}_{k} = (x^2 - \xi X^2 + \nu X^4) + (y^2 + \omega Y^2 + \nu
Y^4 - 2 \nu X^2 Y^2),
\end{equation}
where $(x,X,y,Y)$ are canonically conjugate coordinates, $\xi = - (
a_{13} a_{34} k_1^2 + a_{23} k_2 \sqrt{a_{12} a_{34}}) $, $\omega =
\xi + 2 a_{13} a_{34} k_1^2 = a_{13} a_{34} k_1^2 - a_{23} k_2
\sqrt{a_{12} a_{34}}$.

The Casimirs may be rescaled to $c_1 = \sqrt{a_{13} a_{34}} k_1$ and
$c_2 = a_{23} k_2 \sqrt{a_{12} a_{34}}$. In this case, $\xi =
-c_1^2-c_2$ and $\omega = c_1^2 - c_2$. The ratio $\frac{\omega}{\xi}$
is negative when $\xi < 0 < \omega$, that is, when $c_1^2 > c_2 >
-c_1^2$. Otherwise, either $\xi < \omega < 0$ or $0 < \xi <
\omega$. If $e > 0$, the energy level $\set{ {\bf h}_k = e }$
intersects the set S of Casimir values where the ratio
$\frac{\omega}{\xi} > 0$, i.e. where the origin is a saddle-centre for
${\bf h}_k$ (figure \ref{fig:saddle-center}). In the degenerate case
where $a_{13} = 0$, $c_1 \equiv 0$ and the ratio $\omega/\xi \equiv
1$, except when $c_2=0$, where it is undefined.

\begin{figure}[h]
  \centering
  \caption{The regions in the space of rescaled Casimir values where
    \eqref{eq:h1} has saddle-centre (S) and centre-centre (C)
    equilibria at the origin.}
  \label{fig:saddle-center}
  \resizebox{12cm}{!}{
    \input{saddle-center.ltx}
  }
\end{figure}
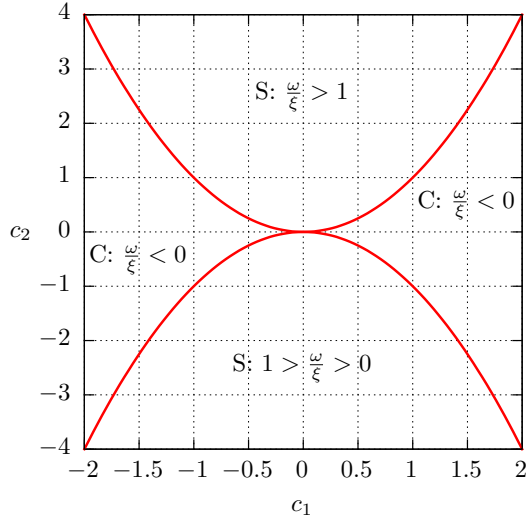

\subsubsection{The geodesic flow and Euler equations} \label{sssec:ham}

When $\frac{\omega}{\xi} > 0$, there is a second change of variables
that transforms ${\bf h}_k$ into a constant multiple of
\begin{equation} \label{eq:h2}
2{\bf h} = x^2 + \left(X^2-\frac{1}{2}\right)^2 + y^2 + \alpha^2 Y^2 +
Y^4 - 2X^2Y^2.
\end{equation}

For all $\epsilon > 0$, the rescaling $(y,Y) \mapsto
(y,Y)/\sqrt{\epsilon}$ transforms the hamiltonian vector-field of
${\bf h}$ (equation \ref{eq:h2}) to the non-hamiltonian vector-field
\begin{equation} \label{eq:vf}
\vf{X}_{\epsilon} = \left\{ 
\begin{array}{lclclcl}
\dot{X} &=& x, & \hspace{2mm} & \dot{Y} &=& y,\\
\dot{x} &=& X - 2X^3 + 2\epsilon XY^2, & & \dot{y} &=& \left[ -\alpha^2 + 2 X^2 \right]\, Y - 2\epsilon Y^3.
\end{array} 
\right.
\end{equation}

\subsubsection{The normally hyperbolic manifold $S$}
The plane
$$S = \set{ (x,X,y,Y)\ \st x=X=0 }$$
is $\vf{X}_{\epsilon}$-invariant for all $\epsilon$. As shown in
\cite{MR2425326}, $S$ is normally-hyperbolic for all $\epsilon$.

\subsubsection{The stable and unstable manifolds of $S$}
\label{ssec:wpmS}

The function $h : T^* \R^2 \to \R$, $h = x^2 + (X^2-\frac{1}{2})^2$ is
a first integral of $\vf{X}_0$. The set $h^{-1}(\frac{1}{4})$ is the
stable and unstable manifold of $S$, which we denote by
$W^{\pm}_0(S)$. On $W^{\pm}_0(S) - S$, the flow of $\vf{X}_0$
satisfies
\begin{align} \label{eq:separatrix}
\begin{array}{lcllcl}
X &=& \pm \sech(t+t_0), & x &=& \mp
\tanh(t+t_0)\sech(t+t_0),\\
Y &=&
c_0 Y_0(t+t_0) + c_1Y_1(t+t_0), & y &=& \dot{Y},
\end{array}
\end{align}
where $X(0)=\pm \sech(t_0), x(0)=\mp \tanh(t_0)\sech(t_0)^2$ and
$\set{Y_j}$ is a basis of solutions to the initial-value problem
\begin{align} \label{eq:deY}
  \ddot{Y} + \left[ \alpha^2 - 2\sech(t)^2 \right]\, Y &= 0,
  && \text{ such that }
  Y(0) = Y_{0,j}, \dot{Y}(0)=\dot{Y}_{0,j} \notag\\
  &&& \text{ and }
  W = Y_{0,0}\dot{Y}_{0,1}-Y_{0,1}\dot{Y}_{0,0} \neq 0
\end{align}
while $Y(0)=c_0Y_0(t_0) + c_1Y_1(t_0)$, $y(0)=c_0\dot{Y}_0(t_0) +
c_1\dot{Y}_1(t_0)$. The particular choice of basis is discussed in section
\ref{sec:scatter+split}.

Given a basis of solutions, this determines a coordinate system
$(t_0,c_0,c_1)$ on the stable and unstable manifolds $W^{\pm}_0(S)$ of
$S$ for $\vf{X}_0$.

\subsubsection{The Melnikov function}
\label{ssec:mel}
By well-known results \cite{MR0501173}, the perturbed stable and
unstable manifolds of $S$, $W^{\pm}_{\epsilon}(S)$ for
$\vf{X}_{\epsilon}$ are, on compact sets, graphs over
$W^{\pm}_0(W)$. The Melnikov function $m$ measures the $O(\epsilon)$
separation of these graphs. In this case, it is a quadratic form in
the coordinates $c_0,c_1$ \cite{MR2425326}:
\begin{align}
m(t_0,c_0,c_1) &= m_{00} c_0^2 + 2m_{01}c_0c_1 + m_{11} c_1^2 &&\text{where}\label{eq:melfn}\\
m_{ij}         &= \int_{\tau \in \R} \dot{q}(\tau)\, Y_i(\tau)\, Y_j(\tau)\ \d\tau, &&
\text{ and }
q(\tau) =  2\sech(\tau)^2.
\label{eq:mij}
\end{align}

\section{Improper integrals via scattering}
\label{sec:improper-integral}

In \cite{MR2425326}, the coefficients $m_{ij}$ are computed in terms
of the asymptotic phase angle between an even and odd solution to
\eqref{eq:deY}. This section examines an alternative route to
computing the coefficients $m_{ij}$ for a general class of scattering
problems and integrals like those in \eqref{eq:mij}.

\begin{definition}
  \label{de:asymptotic}
  Two functions $f,g \in C^1(\R)$ are said to be
  \textem{asymptotically equal} at $+\infty$, written $f \asymp{+} g$,
  if, for each $\epsilon>0$, there is a $T>0$ such that
\begin{align} \label{al:asymptotic}
t \geq T &&\implies&& |f(t)-g(t)| + |f'(t)-g'(t)| < \epsilon.
\end{align}
The definition of asymptotic equality at $-\infty$ is similar and
denoted by $\asymp{-}$.
\end{definition}

\begin{maximacode}
depends([q,w0,w1,w],t);
qdot : diff(q,t);
de : diff(w,t,2)=-(alpha^2 - q)*w;
ibp(u,dv,t) := block([v,du],v:integrate(dv,t),du:diff(u,t),subst(t=0,u*v)-integrate(v*du,t));
\end{maximacode}

Let $q \in C^1(\R) \cap L^1(\R)$ and $\alpha > 0$. Since $q
\asymp{\pm} 0$, there are solutions $w^{\pm}_j$, $j\in\set{0,1}$, to
the differential equation
\begin{align}
\ddot{w} + [\alpha^2 - q(t)]\, w &= 0 \label{eq:dez}
\shortintertext{such that}
w^{\pm}_j(t) &\asymp{\pm} \exp((-1)^j i \alpha t). \label{al:dez-asymptotic}
\end{align}
Given two solutions $w_0,w_1$ to \eqref{eq:dez}--which are not
necessarily the solutions \eqref{al:dez-asymptotic}--, let
\begin{equation} \label{eq:Jpm}
J_{\pm} = \lim_{\pm t \to \infty} \dot{w}_0(t)\dot{w}_1(t)+\alpha^2 w_0(t)w_1(t).
\end{equation}
Since, for each $\sigma\in\set{\pm}$, $\set{w^\sigma_0,w^\sigma_1}$ is a
basis of the solution space to \eqref{eq:dez}, there are constants
$a^{\sigma}_{ij}$ such that
\begin{equation} \label{eq:connection}
\begin{bmatrix}
w_0\\w_1
\end{bmatrix}
=
\begin{bmatrix}
a^{\sigma}_{00} & a^{\sigma}_{01} \\
a^{\sigma}_{10} & a^{\sigma}_{11}
\end{bmatrix}
\,
\begin{bmatrix}
w^{\sigma}_0\\w^{\sigma}_1
\end{bmatrix}.
\end{equation}

\begin{lemma} \label{le:Jpm}
The limits \eqref{eq:Jpm} exist and are equal to
\begin{align} \label{eq:jpm-limit}
J_{\sigma} &= 2\,\alpha^2\, \left[ a^{\sigma}_{01}
  a^{\sigma}_{10} + a^{\sigma}_{00} a^{\sigma}_{11} \right],&& \text{for }\sigma\in\set{\pm}.
\end{align}
\end{lemma}
\begin{proof}
Let $\sigma$ be $+$ or $-$. One sees that
\begin{align}
&\phantom{=\ \ \ }
\dot{w}_0\dot{w}_1+\alpha^2 w_0w_1\notag\\
&=
(a^{\sigma}_{00}\dot{w}^{\sigma}_0+a^{\sigma}_{01}\dot{w}^{\sigma}_1)(a^{\sigma}_{10}\dot{w}^{\sigma}_0+a^{\sigma}_{11}\dot{w}^{\sigma}_1)+
\alpha^2\,(a^{\sigma}_{00}w^{\sigma}_0+a^{\sigma}_{01}w^{\sigma}_1)(a^{\sigma}_{10}w^{\sigma}_0+a^{\sigma}_{11}w^{\sigma}_1)\notag\\
&\asymp{\sigma}
-\alpha^2\,(a^{\sigma}_{00}\exp(i \alpha t)-a^{\sigma}_{01}\exp(-i \alpha t))(a^{\sigma}_{10}\exp(i \alpha t)-a^{\sigma}_{11}\exp(-i \alpha t))\notag\\
&\phantom{=\ \ \ }+\alpha^2\,(a^{\sigma}_{00}\exp(i \alpha t)+a^{\sigma}_{01}\exp(-i \alpha t))(a^{\sigma}_{10}\exp(i \alpha t)+a^{\sigma}_{11}\exp(-i \alpha t))\notag
\end{align}
which implies \eqref{eq:jpm-limit}.
\end{proof}

\begin{theorem} \label{thm:I0}
Let $w_0,w_1$ be solutions to \eqref{eq:dez}. The integral
\begin{equation} \label{eq:Iz}
I = \int_{-\infty}^{\infty} \dot{q}(t)\, w_0(t)\, w_1(t)\, \d t
\end{equation}
exists and equals $J_--J_+$, that is,
\begin{equation} \label{eq:I}
I = 2\, \alpha^2\, \left[ a^-_{01} a^-_{10} + a^-_{00} a^-_{11} -
  a^+_{01} a^+_{10} - a^+_{00} a^+_{11} \right].
\end{equation}
\end{theorem}

\begin{proof}
  The proof is similar to that in~\cite{MR2425326}. Let $I_{\sigma} =
  \sigma\, \int_0^{\sigma\infty} \dot{q}(t)\, w_0(t)\, w_1(t)\, \d t$
  for $\sigma\in\set{\pm}$, so that $I=I_++I_-$. Integration by parts
  shows that $I_+ = -J_+ + C$ and $I_- = J_--C$ where $C = [\alpha^2 +
  q(0)]\,w_0(0)w_1(0)+\dot{w}_0(0)\dot{w}_1(0)$ is a boundary datum.
\end{proof}

The following is useful in computing the Melnikov coefficients
$m_{ij}$ \eqref{eq:mij}.

\begin{corollary} \label{co:I}
Let $I_{ij}$ denote the integral \eqref{eq:Iz} with $w_0 = w^-_i$ and
$w_1 = w^-_j$ for $i,j\in\set{0,1}$. Then
\begin{align}
I_{ii} &= -4 \alpha^2 a^+_{i0} a^+_{i1} && 
I_{01} = I_{10} = 2 \alpha^2 (1 - a^+_{00} a^+_{11} - a^+_{01} a^+_{10})
\end{align}
where $[a^+_{ij}]$ is the connection matrix
\eqref{eq:connection}. Moreover
\begin{equation} \label{eq:detI}
  \det [I_{ij}] = -4\, \alpha^4\,
  \left(
    (1 - a^+_{00} a^+_{11} - a^+_{01} a^+_{10})^2
    - 4 a^+_{00} a^+_{11} a^+_{01} a^+_{10}
  \right).
\end{equation}
\end{corollary}

\section{The scattering coefficients and splitting of the invariant
  manifolds}
\label{sec:scatter+split}

To compute the $m_{ij}$ in \eqref{eq:mij}, it is useful to transform
the differential equation \eqref{eq:deY} into a form that reveals its
solubility in terms of hypergeometric functions. Substitution of
$z=\tanh(t)$ transforms the differential equation \eqref{eq:deY} into
the Legendre differential equation \cite[p. 324]{MR0178117}
\begin{equation} \label{eq:legendre}
(1-z^2)Y'' -2zY' + \left( \nu(\nu+1) - \frac{\mu^2}{1-z^2}
  \right) Y = 0,
\end{equation}
where $\mu=i\alpha$, $\nu=-\frac{1}{2}+\frac{\sqrt{-7}}{2}$ and
$'=\frac{\d\ }{\d z}$.

Let $F(a,b;c;z)$ be the $(2,1)$ hypergeometric function with moduli
$a,b,c \in \C$ and argument $z \in \C$ \cite[p. 281]{MR0178117},
\cite[\S 15.2.1]{MR2655355}. There are four privileged solutions to
\eqref{eq:legendre} that are expressed in terms of these
hypergeometric functions, namely,
\begin{align}
Y^+_0 &= \left[ \frac{1+z}{1-z}  \right]^{\frac{i \alpha}{2}}\, F(a,b;c;\frac{1-z}{2}) 
&&& Y^+_1 = \left[ \frac{1+z}{1-z}  \right]^{-\frac{i \alpha}{2}}\, F(a,b;\bar{c};\frac{1-z}{2})\phantom{,} \notag\\
Y^-_0 &= \left[ \frac{1-z}{1+z}  \right]^{\frac{i \alpha}{2}}\, F(a,b;c;\frac{1+z}{2}) 
&&& Y^-_1 = \left[ \frac{1-z}{1+z}  \right]^{-\frac{i \alpha}{2}}\, F(a,b;\bar{c};\frac{1+z}{2}), \label{al:sleg}
\end{align}
where $a=\frac{1}{2}+\frac{\sqrt{-7}}{2}$, $b=\bar{a}$ and $c=1-i
\alpha$ \cite[p. 286]{MR0178117}. (The notation is explained thus: the
group $\Z_2 \oplus \Z_2$ acts by automorphisms of \eqref{eq:legendre}
by changing the sign of $\alpha$ and $z$ independently.) From the fact
that as $\pm t \to \infty$, $\pm z \to 1$ and $F(a,b;c;0)=1$, it is
apparent that
\begin{align}
Y^+_0 &\asymp{+} \exp(i \alpha t)
&&& Y^+_1 \asymp{+} \exp(-i \alpha t)\phantom{,}\notag\\
Y^-_0 &\asymp{-} \exp(i \alpha t)
&&& Y^-_1 \asymp{-} \exp(-i \alpha t),\label{al:sleg-asymp}
\end{align}
viewed as functions of $t=\tanh^{-1}(z)$.

The linear transformation rules for hypergeometric functions
\cite[15.3.3, 15.3.6]{MR1225604} imply the relations
\begin{align}
Y^-_0 &= A Y^+_0 + B Y^+_1, && Y^-_1 =  \bar{B} Y^+_0 + \bar{A} Y^+_1 &\text{where}\notag\\
A &= \frac{\Gamma(c)\,\Gamma(1-c)}{\Gamma(a)\,\Gamma(b)}
&& B = \frac{\Gamma(c)\,\Gamma(c-1)}{\Gamma(c-a)\,\Gamma(c-b)} \label{al:sleg-scatter}\\
\intertext{so the connection matrices are}
[a^-_{ij}] &= 
\begin{bmatrix}
1&0\\0&1
\end{bmatrix}
&&
[a^+_{ij}] = 
\begin{bmatrix}
A&B\\\bar{B}&\bar{A}
\end{bmatrix}.
\end{align}

\begin{lemma} \label{le:agt}
  Let $s=\sqrt{7}/2$. Then
\begin{align} \label{eq:agt}
  |B|^2 &= \frac{\cosh(2\pi \alpha)+\cosh(2\pi s)}{\cosh(2\pi \alpha)-1},
  &&&
  |A|^2 &= \left( \frac{\cosh(\pi s)}{\sinh(\pi\alpha)} \right)^2.
\end{align}
So, $|B|$ exceeds unity for all $\alpha \in \R, \alpha \neq 0$ and
$|A|/|B|$ is maximized at $\alpha = 0$ and decreases monotonically to
$0$ as $\alpha \to \infty$.
\end{lemma}

\begin{proof}
  Assume that $x \neq 0$. The reflection formula for the
  $\Gamma$-function implies that $|\Gamma(\frac{1}{2}+i x)|^2 =
  \dfrac{\pi}{\cosh(\pi x)}$, $|\Gamma(i x)|^2 = \dfrac{\pi}{|x
    \sinh(\pi x)|}$ and $|\Gamma(1+i x)|^2 = \dfrac{\pi
    |x|}{|\sinh(\pi x)|}$ \cite[{\paragraphsym} 6.1.17,
  6.1.29--31]{MR1225604}. Then, since $c=1-i \alpha$ and
  $c-a=\frac{1}{2}+i(-\alpha-s)$,
\begin{align}
|B|^2 &= \left| 
\dfrac{
  \Gamma(1-i \alpha)\,
  \Gamma(-i \alpha)}{
  \Gamma(\frac{1}{2}+i(-\alpha-s))\,
  \Gamma(\frac{1}{2}+i(-\alpha+s)) }
\right|^2 \notag\\
&= \dfrac{
  \cosh(\pi(-\alpha-s))\,
  \cosh(\pi(-\alpha+s))}{
  \sinh(\pi \alpha)^2} \notag
\end{align}
which yields the first part of \eqref{eq:agt} and implies $|B| \geq 1$
and $>1$ if $\alpha \neq 0$. A similar computation shows the second
part. This implies that
\[
|A|^2/|B|^2 =
\cosh(\pi s)^2 / \left( \cosh(\pi s)^2 + \cosh(\pi \alpha)^2 - 1 \right) \leq 1
\]
and $<1$ when $\alpha \neq 0$ and decreases monotonically as $\alpha
\to \infty$.
\end{proof}

\begin{theorem}
  \label{thm:melnikov-coefficients}
  The Melnikov form with coefficients $m_{ij}$ \eqref{eq:mij} for the
  basis $Y_j = Y^-_j$ is non-degenerate and indefinite for all $\alpha
  \in \R, \alpha \neq 0$.
\end{theorem}

\begin{proof}
When corollary \ref{co:I} is applied, with the connection coefficients
in \eqref{al:sleg-scatter}, one computes that $m_{ij}=I_{ij}$ so
\begin{equation} \label{eq:detmij}
\det [m_{ij}] = -4\, \alpha^4\, ( (1-|A|^2-|B|^2)^2 - 4|A|^2 |B|^2).
\end{equation}
By hypothesis, $\alpha \neq 0$, so $\det [m_{ij}]$ vanishes iff $|A|
\pm |B| = \pm1$. By lemma \ref{le:agt}, the only possible equation to
be satisfied is $|A|=|B|-1$. If this latter equation is satisfied,
then
\[
\left[ \cosh(\pi s) + \sinh(\pi \alpha) \right]^2 = \cosh(\pi s)^2 + \cosh(\pi \alpha)^2 -1,
\]
which implies that $\alpha = 0$. Therefore, $\det [m_{ij}] \neq 0$ for
all real, non-zero $\alpha$.

Indefiniteness of the Melnikov form follows from the even-ness of the
potential $q(t) = -2 \sech(t)^2$: non-trivial even and odd solutions
to \eqref{eq:deY} exist and the Melnikov form vanishes on these
solutions by Theorem \ref{thm:I0}.
\end{proof}

\begin{remark}
  \label{re:comparison}
  Let us compare the results for the Melnikov form \eqref{eq:mij} to
  that obtained in \cite{MR2425326}. The Melnikov form $[m_{ij}]$
  relative to the basis $\set{Y^+_0,Y^+_1}$ has been computed to be
\begin{equation} \label{eq:mij-c}
M_y = -2 \alpha^2
\begin{bmatrix}
 2 A B         & |B|^2+|A|^2-1\\
 |B|^2+|A|^2-1 & 2 \bar{A} \bar{B}
\end{bmatrix}.
\end{equation}
Let $W_0,W_1$ be a pair of solutions whose Wronskian matrix is the
identity at $z=0$; in particular, $W_0$ (resp. $W_1$) is an even
(resp. odd) solution. The Melnikov form in this basis is equal to
\begin{align}
  M_w &=
  \underbrace{2 \alpha \cot(\phaseangle)}_{I(\alpha)} \times
  \begin{bmatrix}
    0 & 1\\1 & 0
  \end{bmatrix}.
\end{align}
The change of variables formula for quadratic forms implies that $\det
M_y = \left( -2 i \alpha \right)^2 \times \det M_w$, where the first
term on the right arises from the Wronskian of $\set{Y^+_0, Y^+_1}$.
This implies that
\begin{equation}
  \label{eq:cotbetasq}
  4\cot(\phaseangle)^2 = 4 |A|^2 |B|^2 - \left( 1 - |A|^2 - |B|^2 \right)^2.
\end{equation}
From this, and lemma \ref{le:agt}, one can numerically compute the
phase angle $\phaseangle$ and the integral $I$ as functions of
$\alpha$. These are depicted in figure \ref{sfig:I-B}. In \cite[Figure
1]{MR2425326}, these quantities were determined by numerically solving
the initial-value problem \ref{eq:deY}. The absolute and relative
errors between the closed form solutions from equation
\eqref{eq:cotbetasq} and the numerical approximations in
\cite{MR2425326} are depicted in figure \ref{fig:error}. This figure
shows the approximations are extremely good, with a mean absolute
error of approximately $1.6 \times 10^{-9}$.
\footnote{The angle $B$ reported in \cite[Figure 1]{MR2425326} equals
  $\pi+\phaseangle$ in the present paper.} One can also show
that $I(0^+) = -2 \cosh(\sqrt{7} \pi /2 )/\pi \cong -20.317$, which is
in close agreement with figure \ref{sfig:I-B}.
\end{remark}

\begin{figure}
  \subfigure[\label{sfig:fs}]{%
    \tiny
    \input{phase.ltx}%
  }%
  \subfigure[\label{sfig:I-B}]{%
    \tiny
    \input{I.ltx}%
  }
  \label{fig:1}
  \caption{~\ref{sfig:fs}, Fundamental solutions to the differential
    equations \eqref{eq:deY} with $\alpha=1$. \ref{sfig:I-B}, Left
    axis: the integral $I$ \eqref{eq:I} as a function of $\alpha$ with
    the solutions in figure \ref{sfig:fs}. \ref{sfig:I-B} Right axis:
    the asymptotic phase angle $\phaseangle$ between even \& odd solutions
    $Y_0$ \& $Y_1$ of
    \eqref{eq:deY}.}
\end{figure}
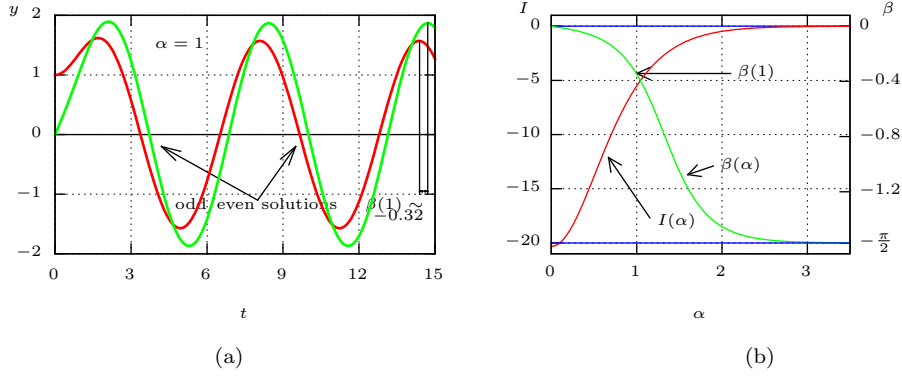

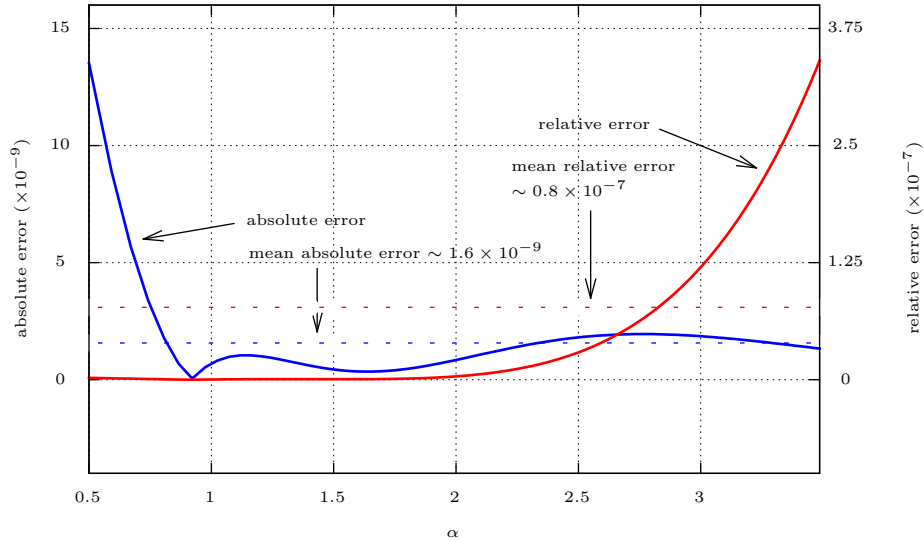
\begin{figure}[h]
  \centering
  \caption{Error in numerical integration of $I(\alpha)$ vs. closed
    form.}
  \tiny
  \input{error.ltx}
  \label{fig:error}
\end{figure}

\bibliographystyle{amsplain}
\bibliography{draft}

\end{document}

%% file: definitions.tex
\usepackage{amsmath,graphicx,ifthen,comment}
\newcommand{\useshowkeys}{0}
\ifthenelse{\equal{\useshowkeys}{1}}{\usepackage{showkeys}}{}
\usepackage{rotating,wrapfig,mathtools}
\usepackage{subfigure}
\usepackage{hyperref}

\renewcommand{\d}{{\rm d}}
\renewcommand{\t}{{\mathfrak t}}

\renewcommand{\asymp}[1]{\mathrel{\inasymp{#1}}}
\newcommand{\inasympop}[1]{%
  \setbox1=\hbox{$\sim$}%
  \setbox2=\hbox{\hspace{-2.25mm}${}_{{}_{#1}}$}%
  \hbox{\raise .1mm\box1\lower .1mm\box2}}
\newcommand{\inasymp}[1]{\hspace{.1em}\mathrel{\inasympop{#1}}\hspace{.1em}}
\newcommand{\coorbit}[0]{{\mathcal O}}
\newcommand{\paragraphsym}{\P}

\newcommand{\Aut}[1]{{\rm Aut}(#1)}
\newcommand{\phaseangle}[0]{\beta}
\newcommand{\sech}{\,{\rm sech}}

\newcommand{\g}{{\mathfrak g}}
\newcommand{\R}{{\bf R}}
\newcommand{\C}{{\bf C}}
\newcommand{\Z}{{\bf Z}}
\newcommand{\ad}[1]{{\rm ad}_{#1}}
\newcommand{\Ad}[1]{{\rm Ad}_{#1}}
\newcommand{\vf}[1]{{\mathcal #1}}

\newcommand{\sbullet}[0]{\mathop{\lower0.5ex\hbox{\Huge$\cdot$}}}
\newcommand{\pb}[2]{\bilinearmap{#1}{#2}{\{}{\}}}
\newcommand{\ip}[2]{\bilinearmap{#1}{#2}{\langle}{\rangle}}
\newcommand{\bilinearmap}[4]{\ifthenelse{\equal{#1}{}}{%
    \bilinearmap{\sbullet}{#2}{#3}{#4}%
  }{%
    \ifthenelse{\equal{#2}{}}{%
      \bilinearmap{#1}{\sbullet}{#3}{#4}%
      }{%
        \left#3#1,#2\right#4
      }}}
\newcommand{\st}[0]{\,:\,}

\newcommand{\set}[1]{\left\{#1\right\}}

\ifcsname theorem\endcsname
\relax\else
\newtheorem{theorem}{Theorem}[section]
\newtheorem{corollary}{Corollary}[section]

\newtheorem{question}{Question}[section]
\newtheorem{lemma}{Lemma}[section]
\newtheorem{definition}{Definition}[section]

\newcounter{remark}
\newenvironment{remark}[1][]{\refstepcounter{remark}\par\medskip\noindent\ignorespaces%
\textbf{Remark~\theremark. #1}\rmfamily}{\medskip\ignorespacesafterend}
\fi

\newcommand{\safeinput}[1]{\IfFileExists{#1}{{\catcode`\&=0\input{#1}}}{\message{File
        #1 does not exists. Skipping.}}}
\newcommand{\textem}[1]{{\em #1}}

%% file: saddle-center.ltx
\begingroup
  \makeatletter
  \providecommand\color[2][]{%
    \GenericError{(gnuplot) \space\space\space\@spaces}{%
      Package color not loaded in conjunction with
      terminal option `colourtext'%
    }{See the gnuplot documentation for explanation.%
    }{Either use 'blacktext' in gnuplot or load the package
      color.sty in LaTeX.}%
    \renewcommand\color[2][]{}%
  }%
  \providecommand\includegraphics[2][]{%
    \GenericError{(gnuplot) \space\space\space\@spaces}{%
      Package graphicx or graphics not loaded%
    }{See the gnuplot documentation for explanation.%
    }{The gnuplot epslatex terminal needs graphicx.sty or graphics.sty.}%
    \renewcommand\includegraphics[2][]{}%
  }%
  \providecommand\rotatebox[2]{#2}%
  \@ifundefined{ifGPcolor}{%
    \newif\ifGPcolor
    \GPcolortrue
  }{}%
  \@ifundefined{ifGPblacktext}{%
    \newif\ifGPblacktext
    \GPblacktexttrue
  }{}%
  \let\gplgaddtomacro\g@addto@macro
  \gdef\gplbacktext{}%
  \gdef\gplfronttext{}%
  \makeatother
  \ifGPblacktext
    \def\colorrgb#1{}%
    \def\colorgray#1{}%
  \else
    \ifGPcolor
      \def\colorrgb#1{\color[rgb]{#1}}%
      \def\colorgray#1{\color[gray]{#1}}%
      \expandafter\def\csname LTw\endcsname{\color{white}}%
      \expandafter\def\csname LTb\endcsname{\color{black}}%
      \expandafter\def\csname LTa\endcsname{\color{black}}%
      \expandafter\def\csname LT0\endcsname{\color[rgb]{1,0,0}}%
      \expandafter\def\csname LT1\endcsname{\color[rgb]{0,1,0}}%
      \expandafter\def\csname LT2\endcsname{\color[rgb]{0,0,1}}%
      \expandafter\def\csname LT3\endcsname{\color[rgb]{1,0,1}}%
      \expandafter\def\csname LT4\endcsname{\color[rgb]{0,1,1}}%
      \expandafter\def\csname LT5\endcsname{\color[rgb]{1,1,0}}%
      \expandafter\def\csname LT6\endcsname{\color[rgb]{0,0,0}}%
      \expandafter\def\csname LT7\endcsname{\color[rgb]{1,0.3,0}}%
      \expandafter\def\csname LT8\endcsname{\color[rgb]{0.5,0.5,0.5}}%
    \else
      \def\colorrgb#1{\color{black}}%
      \def\colorgray#1{\color[gray]{#1}}%
      \expandafter\def\csname LTw\endcsname{\color{white}}%
      \expandafter\def\csname LTb\endcsname{\color{black}}%
      \expandafter\def\csname LTa\endcsname{\color{black}}%
      \expandafter\def\csname LT0\endcsname{\color{black}}%
      \expandafter\def\csname LT1\endcsname{\color{black}}%
      \expandafter\def\csname LT2\endcsname{\color{black}}%
      \expandafter\def\csname LT3\endcsname{\color{black}}%
      \expandafter\def\csname LT4\endcsname{\color{black}}%
      \expandafter\def\csname LT5\endcsname{\color{black}}%
      \expandafter\def\csname LT6\endcsname{\color{black}}%
      \expandafter\def\csname LT7\endcsname{\color{black}}%
      \expandafter\def\csname LT8\endcsname{\color{black}}%
    \fi
  \fi
  \setlength{\unitlength}{0.0500bp}%
  \begin{picture}(7200.00,4320.00)%
    \gplgaddtomacro\gplbacktext{%
      \csname LTb\endcsname%
      \put(1907,595){\makebox(0,0)[r]{\strut{}$-4$}}%
      \csname LTb\endcsname%
      \put(1907,1035){\makebox(0,0)[r]{\strut{}$-3$}}%
      \csname LTb\endcsname%
      \put(1907,1475){\makebox(0,0)[r]{\strut{}$-2$}}%
      \csname LTb\endcsname%
      \put(1907,1915){\makebox(0,0)[r]{\strut{}$-1$}}%
      \csname LTb\endcsname%
      \put(1907,2355){\makebox(0,0)[r]{\strut{}$0$}}%
      \csname LTb\endcsname%
      \put(1907,2795){\makebox(0,0)[r]{\strut{}$1$}}%
      \csname LTb\endcsname%
      \put(1907,3235){\makebox(0,0)[r]{\strut{}$2$}}%
      \csname LTb\endcsname%
      \put(1907,3675){\makebox(0,0)[r]{\strut{}$3$}}%
      \csname LTb\endcsname%
      \put(1907,4115){\makebox(0,0)[r]{\strut{}$4$}}%
      \csname LTb\endcsname%
      \put(2009,409){\makebox(0,0){\strut{}$-2$}}%
      \csname LTb\endcsname%
      \put(2449,409){\makebox(0,0){\strut{}$-1.5$}}%
      \csname LTb\endcsname%
      \put(2889,409){\makebox(0,0){\strut{}$-1$}}%
      \csname LTb\endcsname%
      \put(3329,409){\makebox(0,0){\strut{}$-0.5$}}%
      \csname LTb\endcsname%
      \put(3769,409){\makebox(0,0){\strut{}$0$}}%
      \csname LTb\endcsname%
      \put(4209,409){\makebox(0,0){\strut{}$0.5$}}%
      \csname LTb\endcsname%
      \put(4649,409){\makebox(0,0){\strut{}$1$}}%
      \csname LTb\endcsname%
      \put(5089,409){\makebox(0,0){\strut{}$1.5$}}%
      \csname LTb\endcsname%
      \put(5529,409){\makebox(0,0){\strut{}$2$}}%
      \csname LTb\endcsname%
      \put(1508,2355){\makebox(0,0){\strut{}$c_2$}}%
      \csname LTb\endcsname%
      \put(3769,130){\makebox(0,0){\strut{}$c_1$}}%
      \put(3769,3455){\makebox(0,0){\strut{}S: $\frac{\omega}{\xi}> 1$}}%
      \put(3769,1255){\makebox(0,0){\strut{}S: $1> \frac{\omega}{\xi}> 0$}}%
      \put(5089,2575){\makebox(0,0){\strut{}C: $\frac{\omega}{\xi} < 0$}}%
      \put(2449,2135){\makebox(0,0){\strut{}C: $\frac{\omega}{\xi} < 0$}}%
    }%
    \gplgaddtomacro\gplfronttext{%
    }%
    \gplbacktext
    \put(0,0){\includegraphics{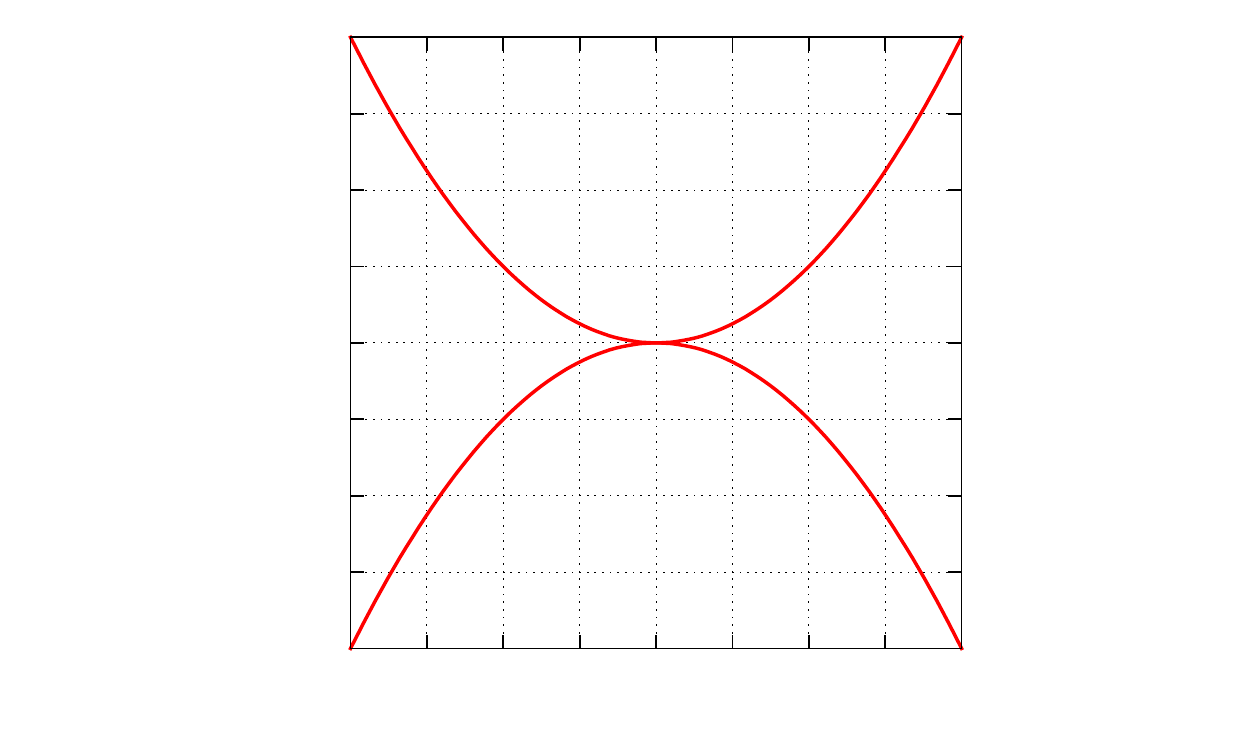}}%
    \gplfronttext
  \end{picture}%
\endgroup

%% file: phase.ltx
\begingroup
  \makeatletter
  \providecommand\color[2][]{%
    \GenericError{(gnuplot) \space\space\space\@spaces}{%
      Package color not loaded in conjunction with
      terminal option `colourtext'%
    }{See the gnuplot documentation for explanation.%
    }{Either use 'blacktext' in gnuplot or load the package
      color.sty in LaTeX.}%
    \renewcommand\color[2][]{}%
  }%
  \providecommand\includegraphics[2][]{%
    \GenericError{(gnuplot) \space\space\space\@spaces}{%
      Package graphicx or graphics not loaded%
    }{See the gnuplot documentation for explanation.%
    }{The gnuplot epslatex terminal needs graphicx.sty or graphics.sty.}%
    \renewcommand\includegraphics[2][]{}%
  }%
  \providecommand\rotatebox[2]{#2}%
  \@ifundefined{ifGPcolor}{%
    \newif\ifGPcolor
    \GPcolortrue
  }{}%
  \@ifundefined{ifGPblacktext}{%
    \newif\ifGPblacktext
    \GPblacktexttrue
  }{}%
  \let\gplgaddtomacro\g@addto@macro
  \gdef\gplbacktext{}%
  \gdef\gplfronttext{}%
  \makeatother
  \ifGPblacktext
    \def\colorrgb#1{}%
    \def\colorgray#1{}%
  \else
    \ifGPcolor
      \def\colorrgb#1{\color[rgb]{#1}}%
      \def\colorgray#1{\color[gray]{#1}}%
      \expandafter\def\csname LTw\endcsname{\color{white}}%
      \expandafter\def\csname LTb\endcsname{\color{black}}%
      \expandafter\def\csname LTa\endcsname{\color{black}}%
      \expandafter\def\csname LT0\endcsname{\color[rgb]{1,0,0}}%
      \expandafter\def\csname LT1\endcsname{\color[rgb]{0,1,0}}%
      \expandafter\def\csname LT2\endcsname{\color[rgb]{0,0,1}}%
      \expandafter\def\csname LT3\endcsname{\color[rgb]{1,0,1}}%
      \expandafter\def\csname LT4\endcsname{\color[rgb]{0,1,1}}%
      \expandafter\def\csname LT5\endcsname{\color[rgb]{1,1,0}}%
      \expandafter\def\csname LT6\endcsname{\color[rgb]{0,0,0}}%
      \expandafter\def\csname LT7\endcsname{\color[rgb]{1,0.3,0}}%
      \expandafter\def\csname LT8\endcsname{\color[rgb]{0.5,0.5,0.5}}%
    \else
      \def\colorrgb#1{\color{black}}%
      \def\colorgray#1{\color[gray]{#1}}%
      \expandafter\def\csname LTw\endcsname{\color{white}}%
      \expandafter\def\csname LTb\endcsname{\color{black}}%
      \expandafter\def\csname LTa\endcsname{\color{black}}%
      \expandafter\def\csname LT0\endcsname{\color{black}}%
      \expandafter\def\csname LT1\endcsname{\color{black}}%
      \expandafter\def\csname LT2\endcsname{\color{black}}%
      \expandafter\def\csname LT3\endcsname{\color{black}}%
      \expandafter\def\csname LT4\endcsname{\color{black}}%
      \expandafter\def\csname LT5\endcsname{\color{black}}%
      \expandafter\def\csname LT6\endcsname{\color{black}}%
      \expandafter\def\csname LT7\endcsname{\color{black}}%
      \expandafter\def\csname LT8\endcsname{\color{black}}%
    \fi
  \fi
  \setlength{\unitlength}{0.0500bp}%
  \begin{picture}(3600.00,2592.00)%
    \gplgaddtomacro\gplbacktext{%
      \csname LTb\endcsname%
      \put(367,595){\makebox(0,0)[r]{\strut{}$-2$}}%
      \csname LTb\endcsname%
      \put(367,1043){\makebox(0,0)[r]{\strut{}$-1$}}%
      \csname LTb\endcsname%
      \put(367,1491){\makebox(0,0)[r]{\strut{}$0$}}%
      \csname LTb\endcsname%
      \put(367,1939){\makebox(0,0)[r]{\strut{}$1$}}%
      \csname LTb\endcsname%
      \put(367,2387){\makebox(0,0)[r]{\strut{}$2$}}%
      \csname LTb\endcsname%
      \put(459,427){\makebox(0,0){\strut{}$0$}}%
      \csname LTb\endcsname%
      \put(1026,427){\makebox(0,0){\strut{}$3$}}%
      \csname LTb\endcsname%
      \put(1593,427){\makebox(0,0){\strut{}$6$}}%
      \csname LTb\endcsname%
      \put(2159,427){\makebox(0,0){\strut{}$9$}}%
      \csname LTb\endcsname%
      \put(2726,427){\makebox(0,0){\strut{}$12$}}%
      \csname LTb\endcsname%
      \put(3293,427){\makebox(0,0){\strut{}$15$}}%
      \csname LTb\endcsname%
      \put(162,2421){\makebox(0,0){\strut{}$y$}}%
      \csname LTb\endcsname%
      \put(1876,130){\makebox(0,0){\strut{}$t$}}%
      \put(1215,2163){\makebox(0,0)[l]{\strut{}$\alpha=1$}}%
      \put(1970,953){\makebox(0,0){\strut{}odd/even solutions}}%
      \put(3218,953){\makebox(0,0)[r]{\strut{}$&phaseangle(1) &sim$}}%
      \put(3218,864){\makebox(0,0)[r]{\strut{}$-0.32$}}%
    }%
    \gplgaddtomacro\gplfronttext{%
    }%
    \gplbacktext
    \put(0,0){\includegraphics{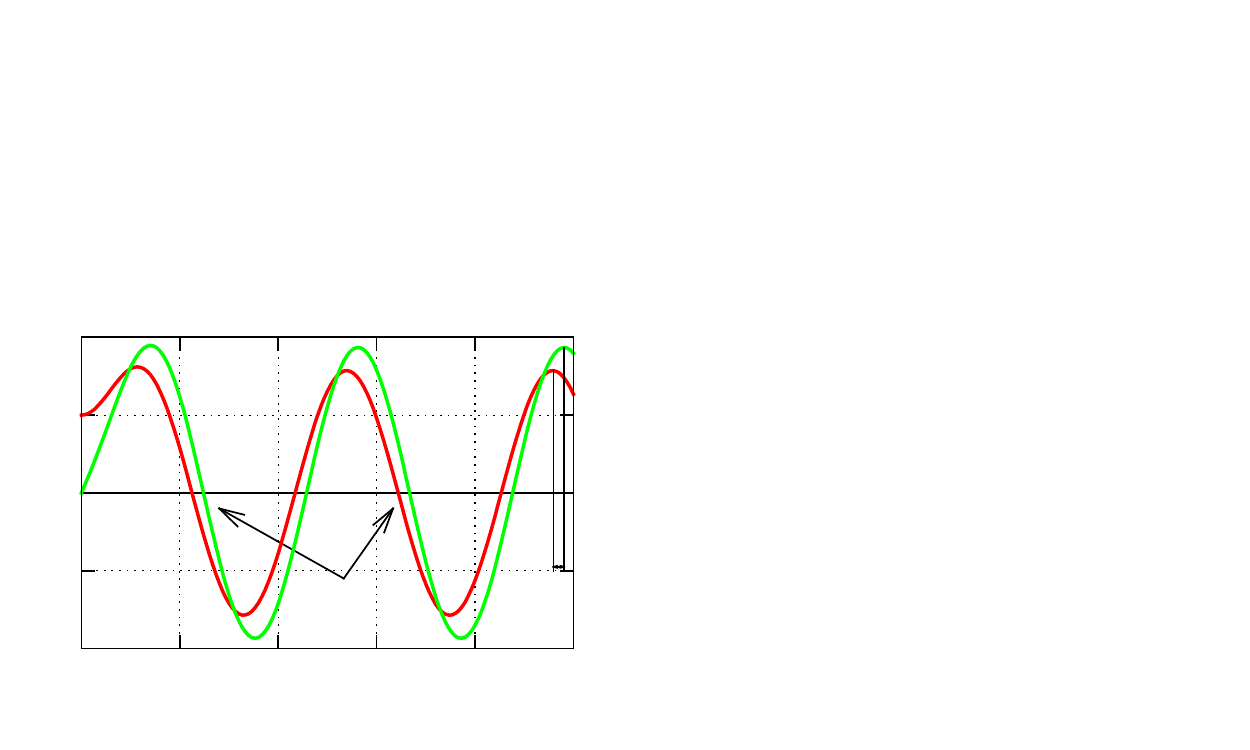}}%
    \gplfronttext
  \end{picture}%
\endgroup

%% file: I.ltx
\begingroup
  \makeatletter
  \providecommand\color[2][]{%
    \GenericError{(gnuplot) \space\space\space\@spaces}{%
      Package color not loaded in conjunction with
      terminal option `colourtext'%
    }{See the gnuplot documentation for explanation.%
    }{Either use 'blacktext' in gnuplot or load the package
      color.sty in LaTeX.}%
    \renewcommand\color[2][]{}%
  }%
  \providecommand\includegraphics[2][]{%
    \GenericError{(gnuplot) \space\space\space\@spaces}{%
      Package graphicx or graphics not loaded%
    }{See the gnuplot documentation for explanation.%
    }{The gnuplot epslatex terminal needs graphicx.sty or graphics.sty.}%
    \renewcommand\includegraphics[2][]{}%
  }%
  \providecommand\rotatebox[2]{#2}%
  \@ifundefined{ifGPcolor}{%
    \newif\ifGPcolor
    \GPcolortrue
  }{}%
  \@ifundefined{ifGPblacktext}{%
    \newif\ifGPblacktext
    \GPblacktexttrue
  }{}%
  \let\gplgaddtomacro\g@addto@macro
  \gdef\gplbacktext{}%
  \gdef\gplfronttext{}%
  \makeatother
  \ifGPblacktext
    \def\colorrgb#1{}%
    \def\colorgray#1{}%
  \else
    \ifGPcolor
      \def\colorrgb#1{\color[rgb]{#1}}%
      \def\colorgray#1{\color[gray]{#1}}%
      \expandafter\def\csname LTw\endcsname{\color{white}}%
      \expandafter\def\csname LTb\endcsname{\color{black}}%
      \expandafter\def\csname LTa\endcsname{\color{black}}%
      \expandafter\def\csname LT0\endcsname{\color[rgb]{1,0,0}}%
      \expandafter\def\csname LT1\endcsname{\color[rgb]{0,1,0}}%
      \expandafter\def\csname LT2\endcsname{\color[rgb]{0,0,1}}%
      \expandafter\def\csname LT3\endcsname{\color[rgb]{1,0,1}}%
      \expandafter\def\csname LT4\endcsname{\color[rgb]{0,1,1}}%
      \expandafter\def\csname LT5\endcsname{\color[rgb]{1,1,0}}%
      \expandafter\def\csname LT6\endcsname{\color[rgb]{0,0,0}}%
      \expandafter\def\csname LT7\endcsname{\color[rgb]{1,0.3,0}}%
      \expandafter\def\csname LT8\endcsname{\color[rgb]{0.5,0.5,0.5}}%
    \else
      \def\colorrgb#1{\color{black}}%
      \def\colorgray#1{\color[gray]{#1}}%
      \expandafter\def\csname LTw\endcsname{\color{white}}%
      \expandafter\def\csname LTb\endcsname{\color{black}}%
      \expandafter\def\csname LTa\endcsname{\color{black}}%
      \expandafter\def\csname LT0\endcsname{\color{black}}%
      \expandafter\def\csname LT1\endcsname{\color{black}}%
      \expandafter\def\csname LT2\endcsname{\color{black}}%
      \expandafter\def\csname LT3\endcsname{\color{black}}%
      \expandafter\def\csname LT4\endcsname{\color{black}}%
      \expandafter\def\csname LT5\endcsname{\color{black}}%
      \expandafter\def\csname LT6\endcsname{\color{black}}%
      \expandafter\def\csname LT7\endcsname{\color{black}}%
      \expandafter\def\csname LT8\endcsname{\color{black}}%
    \fi
  \fi
  \setlength{\unitlength}{0.0500bp}%
  \begin{picture}(4320.00,2592.00)%
    \gplgaddtomacro\gplbacktext{%
      \csname LTb\endcsname%
      \put(489,676){\makebox(0,0)[r]{\strut{}$-20$}}%
      \csname LTb\endcsname%
      \put(489,1084){\makebox(0,0)[r]{\strut{}$-15$}}%
      \csname LTb\endcsname%
      \put(489,1491){\makebox(0,0)[r]{\strut{}$-10$}}%
      \csname LTb\endcsname%
      \put(489,1898){\makebox(0,0)[r]{\strut{}$-5$}}%
      \csname LTb\endcsname%
      \put(489,2306){\makebox(0,0)[r]{\strut{}$0$}}%
      \csname LTb\endcsname%
      \put(561,446){\makebox(0,0){\strut{}$0$}}%
      \csname LTb\endcsname%
      \put(1198,446){\makebox(0,0){\strut{}$1$}}%
      \csname LTb\endcsname%
      \put(1834,446){\makebox(0,0){\strut{}$2$}}%
      \csname LTb\endcsname%
      \put(2471,446){\makebox(0,0){\strut{}$3$}}%
      \csname LTb\endcsname%
      \put(2871,676){\makebox(0,0)[l]{\strut{}$-&frac{&pi}{2}$}}%
      \csname LTb\endcsname%
      \put(2871,1061){\makebox(0,0)[l]{\strut{}$-1.2$}}%
      \csname LTb\endcsname%
      \put(2871,1476){\makebox(0,0)[l]{\strut{}$-0.8$}}%
      \csname LTb\endcsname%
      \put(2871,1891){\makebox(0,0)[l]{\strut{}$-0.4$}}%
      \csname LTb\endcsname%
      \put(2871,2306){\makebox(0,0)[l]{\strut{}$0$}}%
      \csname LTb\endcsname%
      \put(366,2421){\makebox(0,0){\strut{}$I$}}%
      \csname LTb\endcsname%
      \put(3085,2421){\makebox(0,0){\strut{}$&phaseangle$}}%
      \csname LTb\endcsname%
      \put(1675,130){\makebox(0,0){\strut{}$&alpha$}}%
      \put(1363,832){\makebox(0,0)[l]{\strut{}$I(&alpha)$}}%
      \put(1834,1247){\makebox(0,0)[l]{\strut{}$&phaseangle(\alpha)$}}%
      \put(1961,1947){\makebox(0,0)[l]{\strut{}$&phaseangle(1)$}}%
    }%
    \gplgaddtomacro\gplfronttext{%
    }%
    \gplbacktext
    \put(0,0){\includegraphics{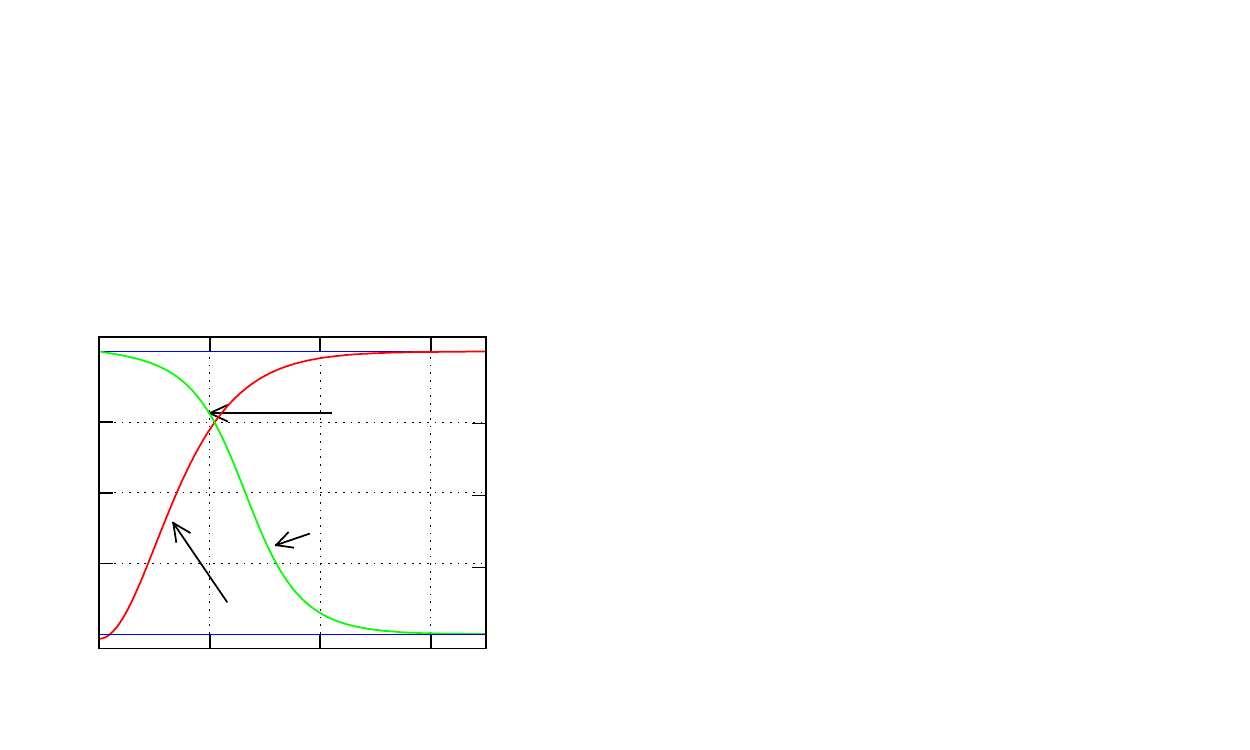}}%
    \gplfronttext
  \end{picture}%
\endgroup

%% file: error.ltx
\begingroup
  \makeatletter
  \providecommand\color[2][]{%
    \GenericError{(gnuplot) \space\space\space\@spaces}{%
      Package color not loaded in conjunction with
      terminal option `colourtext'%
    }{See the gnuplot documentation for explanation.%
    }{Either use 'blacktext' in gnuplot or load the package
      color.sty in LaTeX.}%
    \renewcommand\color[2][]{}%
  }%
  \providecommand\includegraphics[2][]{%
    \GenericError{(gnuplot) \space\space\space\@spaces}{%
      Package graphicx or graphics not loaded%
    }{See the gnuplot documentation for explanation.%
    }{The gnuplot epslatex terminal needs graphicx.sty or graphics.sty.}%
    \renewcommand\includegraphics[2][]{}%
  }%
  \providecommand\rotatebox[2]{#2}%
  \@ifundefined{ifGPcolor}{%
    \newif\ifGPcolor
    \GPcolortrue
  }{}%
  \@ifundefined{ifGPblacktext}{%
    \newif\ifGPblacktext
    \GPblacktexttrue
  }{}%
  \let\gplgaddtomacro\g@addto@macro
  \gdef\gplbacktext{}%
  \gdef\gplfronttext{}%
  \makeatother
  \ifGPblacktext
    \def\colorrgb#1{}%
    \def\colorgray#1{}%
  \else
    \ifGPcolor
      \def\colorrgb#1{\color[rgb]{#1}}%
      \def\colorgray#1{\color[gray]{#1}}%
      \expandafter\def\csname LTw\endcsname{\color{white}}%
      \expandafter\def\csname LTb\endcsname{\color{black}}%
      \expandafter\def\csname LTa\endcsname{\color{black}}%
      \expandafter\def\csname LT0\endcsname{\color[rgb]{1,0,0}}%
      \expandafter\def\csname LT1\endcsname{\color[rgb]{0,1,0}}%
      \expandafter\def\csname LT2\endcsname{\color[rgb]{0,0,1}}%
      \expandafter\def\csname LT3\endcsname{\color[rgb]{1,0,1}}%
      \expandafter\def\csname LT4\endcsname{\color[rgb]{0,1,1}}%
      \expandafter\def\csname LT5\endcsname{\color[rgb]{1,1,0}}%
      \expandafter\def\csname LT6\endcsname{\color[rgb]{0,0,0}}%
      \expandafter\def\csname LT7\endcsname{\color[rgb]{1,0.3,0}}%
      \expandafter\def\csname LT8\endcsname{\color[rgb]{0.5,0.5,0.5}}%
    \else
      \def\colorrgb#1{\color{black}}%
      \def\colorgray#1{\color[gray]{#1}}%
      \expandafter\def\csname LTw\endcsname{\color{white}}%
      \expandafter\def\csname LTb\endcsname{\color{black}}%
      \expandafter\def\csname LTa\endcsname{\color{black}}%
      \expandafter\def\csname LT0\endcsname{\color{black}}%
      \expandafter\def\csname LT1\endcsname{\color{black}}%
      \expandafter\def\csname LT2\endcsname{\color{black}}%
      \expandafter\def\csname LT3\endcsname{\color{black}}%
      \expandafter\def\csname LT4\endcsname{\color{black}}%
      \expandafter\def\csname LT5\endcsname{\color{black}}%
      \expandafter\def\csname LT6\endcsname{\color{black}}%
      \expandafter\def\csname LT7\endcsname{\color{black}}%
      \expandafter\def\csname LT8\endcsname{\color{black}}%
    \fi
  \fi
  \setlength{\unitlength}{0.0500bp}%
  \begin{picture}(7200.00,4320.00)%
    \gplgaddtomacro\gplbacktext{%
      \csname LTb\endcsname%
      \put(441,1299){\makebox(0,0){$0$}}%
      \csname LTb\endcsname%
      \put(441,2179){\makebox(0,0){$5$}}%
      \csname LTb\endcsname%
      \put(441,3059){\makebox(0,0){$10$}}%
      \csname LTb\endcsname%
      \put(441,3939){\makebox(0,0){$15$}}%
      \csname LTb\endcsname%
      \put(645,409){\makebox(0,0){$0.5$}}%
      \csname LTb\endcsname%
      \put(1558,409){\makebox(0,0){$1$}}%
      \csname LTb\endcsname%
      \put(2470,409){\makebox(0,0){$1.5$}}%
      \csname LTb\endcsname%
      \put(3383,409){\makebox(0,0){$2$}}%
      \csname LTb\endcsname%
      \put(4296,409){\makebox(0,0){$2.5$}}%
      \csname LTb\endcsname%
      \put(5208,409){\makebox(0,0){$3$}}%
      \csname LTb\endcsname%
      \put(6299,1299){\makebox(0,0){$0$}}%
      \csname LTb\endcsname%
      \put(6299,2179){\makebox(0,0){$1.25$}}%
      \csname LTb\endcsname%
      \put(6299,3059){\makebox(0,0){$2.5$}}%
      \csname LTb\endcsname%
      \put(6299,3939){\makebox(0,0){$3.75$}}%
      \csname LTb\endcsname%
      \put(144,2355){\rotatebox{-270}{\makebox(0,0){absolute error ($&times 10^{-9}$)}}}%
      \csname LTb\endcsname%
      \put(6799,2355){\rotatebox{-270}{\makebox(0,0){relative error ($&times 10^{-7}$)}}}%
      \csname LTb\endcsname%
      \put(3370,130){\makebox(0,0){$&alpha$}}%
      \put(3370,4022){\makebox(0,0){\strut{}}}%
      \put(1831,2478){\makebox(0,0)[l]{\strut{}absolute error}}%
      \put(4843,3200){\makebox(0,0)[r]{\strut{}relative error}}%
      \put(1848,2241){\makebox(0,0)[l]{\strut{}mean absolute error $&sim 1.6 &times 10^{-9}$}}%
      \put(3806,2883){\makebox(0,0)[l]{\strut{}mean relative error}}%
      \put(3806,2707){\makebox(0,0)[l]{\strut{}$&sim 0.8 &times 10^{-7}$}}%
    }%
    \gplgaddtomacro\gplfronttext{%
    }%
    \gplbacktext
    \put(0,0){\includegraphics{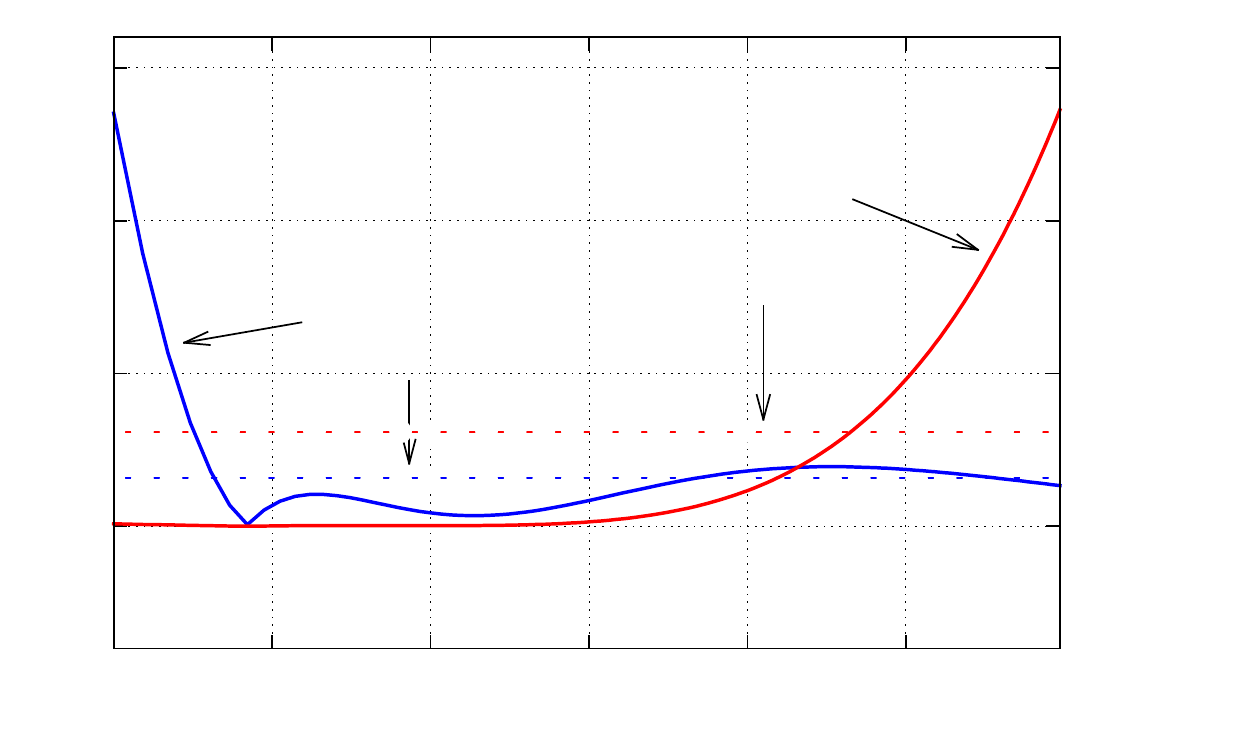}}%
    \gplfronttext
  \end{picture}%
\endgroup